\documentclass[letterpaper, 10 pt, conference]{ieeeconf}  % Comment this line out if you need a4paper

\IEEEoverridecommandlockouts                              % This command is only needed if 
\usepackage{color,graphicx,amssymb,amsmath,verbatim,epsfig,epstopdf,multirow,mathrsfs}
\usepackage{setspace,amstext,array,tabu,bm}   
\usepackage{adjustbox}
\usepackage{url}
\usepackage{slashbox}
\newcolumntype{L}{>{$}l<{$}}  % math-mode version of "l" column type
\usepackage{hyperref}
\usepackage{hypcap,fancyhdr}
 
\newtheorem{theorem}{Theorem}

\newtheorem{example}[theorem]{Example}

\newtheorem{proposition}[theorem]{Proposition}
\newtheorem{remark}[theorem]{Remark}

\newenvironment{mat}{\left(\begin{array}{ccccccccccccccc}}{\end{array}\right)}
\newcommand\bcm{\begin{mat}}
\newcommand\ecm{\end{mat}}

\newenvironment{rmat}{\left[\begin{array}{rrrrrrrrrrrrr}}{\end{array}\right]}
\newcommand\brm{\begin{rmat}}
\newcommand\erm{\end{rmat}}

\newcommand{\EE}{{\mathord{I\kern -.33em E}}}
\def\E{{\EE}} % expectation
\def\Q{{\mathbb Q}} % rationals
\def\P{{\mathbb P}} 
\def\R{{\mathbb R}} % reals
\def\1{{\bm 1}} % indicator
\DeclareMathOperator{\Tr}{Tr}
\DeclareMathOperator{\rank}{rank}
\title{\LARGE \bf  Optimal Dynamic Futures Portfolios Under \\ a Multiscale Central Tendency Ornstein-Uhlenbeck Model}

\author{Tim Leung$^{1}$ and Yang Zhou$^{2}$
\thanks{$^{1}$Department of Applied Mathematics, University of Washington, Seattle WA 98195. E-mail:
\mbox{timleung@uw.edu}. Corresponding author. } \thanks{$^{2}$Department of Applied Mathematics, University of Washington, Seattle WA 98195. E-mail:
\mbox{yzhou7@uw.edu}.} }
\begin{document}

\maketitle
\thispagestyle{empty}
\pagestyle{empty}

%%%%%%%%%%%%%%%%%%%%%%%%%%%%%%%%%%%%%%%%%%%%%%%%%%%%%%%%%%%%%%%%%%%%%%%%%%%%%%%%
\begin{abstract} We study the problem of dynamically trading multiple futures whose   underlying asset price  follows a multiscale central tendency Ornstein-Uhlenbeck (MCTOU) model. Under this model, we derive the closed-form no-arbitrage prices for the futures contracts. Applying a utility maximization approach, we solve for the optimal trading strategies under different portfolio configurations by examining the associated system of Hamilton-Jacobi-Bellman (HJB) equations. The optimal strategies depend on not only the parameters of the underlying asset price process, but also the risk premia embedded in the futures prices.  Numerical examples are provided to illustrate the investor's optimal positions and optimal wealth over time. \end{abstract}

%%%%%%%%%%%%%%%%%%%%%%%%%%%%%%%%%%%%%%%%%%%%%%%%%%%%%%%%%%%%%%%%%%%%%%%%%%%%%%%%
\section{Introduction}

Futures are standardized exchange-traded bilateral \mbox{contracts} of agreement to buy or sell an asset at a pre-determined price at a pre-specified time in the future.  The underlying asset can be a physical commodity, like gold and silver, or oil and gas, but it can also be a  market index like the S\&P 500 index or the CBOE volatility index. 

Futures are an integral part of the derivatives market. The Chicago Mercantile Exchange (CME), which is the world's largest futures exchange, averages over 15 million futures contracts traded per day.\footnote[3]{Source: CME Group daily exchange volume and open interest report, available at \url{https://www.cmegroup.com/market-data/volume-open-interest/exchange-volume.html}} Within the universe of hedge funds and alternative investments, futures funds constitute a major part with hundreds of billions under management. This motivates us to investigate the problem of trading  futures portfolio dynamically over time.
 
In this paper, we introduce a   multiscale central tendency Ornstein-Uhlenbeck (MCTOU) model to describe the  price dynamics of the underlying asset.  This is a  three-factor model  that is driven by  a fast  mean-reverting OU process and a slow  mean-reverting OU process.  Similar multiscale framework has been  widely used for modeling stochastic volatility of stock prices  \cite{Fouque2000}. The flexibility of   multifactor models permits good fit to empirical term structure as displayed in the market. Especially in deep and liquid futures markets, such as crude oil or gold, with over ten contracts of various maturities actively traded at any given time, multifactor models are particular useful.  In the literature,  we refer to \cite{nfactor} for  a multifactor Gaussian model for pricing oil futures, and \cite{CORTAZAR2017182} for a multifactor stochastic volatility model for commodity prices to enhance calibration against observed option prices.

Under the MCTOU model, we first derive the no-arbitrage price formulae for the futures contracts.  In turn, we solve a utility maximization problem to derive the optimal trading strategies over a finite trading horizon.  This stochastic control approach  leads to the analysis of the associated    Hamilton-Jacobi-Bellman (HJB) partial differential equation satisfied by the investor's value function. We derive  both the  investor's value function and   optimal strategy explicitly. 

Our solution also yields the formula for the investor's certainty equivalent, which   quantifies the value of the futures trading opportunity to the investor. Surprisingly the value function, optimal strategy and certainty equivalent   depend not on the current spot and futures prices, but on the associated risk premia. In addition, we provide the numerical examples to illustrate the investor's optimal futures positions and optimal wealth over time.
 
In the literature, stochastic control approach has been widely applied to continuous-time  dynamic optimization of  stock portfolios  dating  back to \cite{Merton}, but much less has been done for portfolios of futures and other derivatives.  For futures portfolios, one must account for the risk-neutral pricing before solving for the optimal trading strategies. To that end, our model falls within the multi-factor Gassian model for futures pricing, as used for oil futures in \cite{nfactor}.  The utility maximization approach is used to derive dynamic futures trading strategies under two-factor models in \cite{firstpaper} and \cite{secondpaper}. A general regime-switching framework for dynamic futures trading can be found in \cite{Leung2019}. As an alternative approach for capturing  futures and spot price dynamics,  the stochastic basis model \cite{BahmanLeung, Angoshtari2019b} directly models the difference between the futures and underlying asset prices, and solve for the optimal trading strategies through utility maximization.

%The stochastic approach to trading cointegrated securities in a general multi-factor setting have been presented in \cite{CarteaJaimungal,Jorge}. 

In comparison to these studies, we have extended the investigation of optimal trading in commodity futures market under two-factor models to a three-factor model. Closed-form expressions for the optimal controls and for the value function are obtained. Using these formulae, we illustrate  the optimal strategies. Intuitively, it should be more beneficial to be able to access a larger set of securities, and this intuition is confirmed quantitatively. Therefore, we consider all available contracts of different maturities in that market. From our numerical example, the highest certainty equivalent is achieved from trading every contract that is available.

There are a few alternative approaches and applications of dynamic futures portfolios. In \cite{LeungLiLiZheng2015} and   \cite{meanreversionbook2016}, an optimal stopping approach for futures trading is studied. In practice, dynamic futures portfolios are also commonly used to track a commodity index \cite{LeungWard}.

%As used in \cite{Schwartz97}, Kalman filtering methodology can handle multi-factor models with hidden state variables and measurement errors. 

%In reality, the critical concern is transaction costs -- trading more contracts might not necessarily be more profitable in the presence of market frictions.
%It is equally important to identify the break-even point in the number of contracts beyond which profitability might be eroded by transaction costs,
%and the determination of which requires incorporation of transaction costs into the model.

%Inclusion of measurement errors is necessary since the number of available market prices is generally higher than the number of state variables that need to be estimated.
%In addition, Kalman filtering is capable of using a large price panel in the estimation process, 
%avoiding the necessity of making an arbitrary selection of contracts to include in the estimation.

%The rest of this paper is structured as follows. We introduce the MCTOU process in Section \ref{sect-Multiscale-CTOU}. The futures pricing and portfolio optimization problem are discussed in Section \ref{sect-formula-multiscale}.   Numerical examples are provided  in Section \ref{sect-numeric}. 

 \section{The Multiscale Central Tendency Ornstein-Uhlenbeck Model}\label{sect-Multiscale-CTOU}
We now present the multiscale central tendency Ornstein-Uhlenbeck (MCTOU) model that describes the price dynamics of the underlying asset. This leads to the no-arbitrage pricing of the associated futures contracts. Hence, the dynamcis under both the physical measure $\mathbb{P}$ and risk-neutral pricing measure $\mathbb{Q}$ are discussed. 

\subsection{Model Formulation}
We fix a probability space $(\Omega, \mathcal{F},\mathbb{P})$. The log-price of the underlying asset $S_t$ is denoted by $X_t^{(1)}$. Its evolution under the physical measure $\P$ is given by the system of stochastic differential equations:
\begin{align}
dX^{(1)}_t &= \kappa\left(X^{(2)}_t+X^{(3)}_t-X^{(1)}_t\right)dt+\sigma_1 dZ^{\P,1}_t,\label{model-multiscale-begin}\\
dX^{(2)}_t &= \frac{1}{\epsilon}\left(\alpha_2-X^{(2)}_t\right)dt\notag\\
&+\frac{1}{\sqrt \epsilon}\sigma_2 \left(\rho_{12} dZ^{\P,1}_t+\sqrt{1-\rho_{12}^2} dZ^{\P,2}_t\right),\label{model-multiscale-fast}\\
dX^{(3)}_t &= \delta\left(\alpha_3-X^{(3)}_t\right)dt+\sqrt{\delta}\sigma_3\bigg(\rho_{13} dZ^{\P,1}_t\notag\\
&+\rho_{23}dZ^{\P,2}_t+\sqrt{1-\rho_{13}^2-\rho_{23}^2}dZ^{\P,3}_t\bigg),\label{model-multiscale-end}
\end{align}
where $Z^{\P,1}$ , $Z^{\P,2}$ and $Z^{\P,3}$ are independent Brownian motions under the physical measure $\P$. 

Under this model, the   mean process of log-price $X^{(1)}$ is the sum of two stochastic factors, $X_t^{(2)}$ and $X_t^{(3)}$, modeled by two different OU processes. The first factor $X_t^{(2)}$ is
 fast mean-reverting. The rate of mean reversion is  represented by  $1/\epsilon$, with $\epsilon > 0$ being a small parameter corresponding to the time scale of this process. $X_t^{(2)}$ is an ergodic process
and its invariant distribution is independent of $\epsilon$. This distribution is Gaussian with mean $\alpha_2$ and variance $\sigma_2^2/2$.   In contrast, the second factor $X_t^{(3)}$ is a slowly mean-reverting OU process. The rate of mean reversion is represented   by a small parameter $\delta>0$. 

The three processes $X^{(1)}_t$, $X_t^{(2)}$, and $X_t^{(3)}$ can be correlated.  The correlation coefficients $\rho_{12}$, $\rho_{13}$, and $\rho_{23}$ are constants, which satisfy $|\rho_{12}|<1$ and $\rho_{13}^2+\rho_{23}^2<1$.

We specify the market prices of risk as $\zeta_i$, for $i=1,2,3$, which satisfy
\begin{equation}
dZ^{\Q,i}_t=dZ_t^{\P,i}+\zeta_i\, dt,
\end{equation}
where $Z^{\Q,1}$, $Z^{\Q,2}$ and $Z^{\Q,3}$ are independent Brownian motions under risk-neutral pricing measure $\Q$. We introduce the combined market prices of   risk $\lambda_i$, for $i=1,2,3$, defined by
\begin{align}
\lambda_1 &=\zeta_1 \sigma_1 ,\\
\lambda_2 &= \frac{1}{\sqrt \epsilon}\sigma_2 \left(\zeta_1\rho_{12} +\zeta_2\sqrt{1-\rho_{12}^2} \right),\\
\lambda_3 &= \sqrt{\delta}\sigma_3\left(\zeta_1\rho_{13} +\zeta_2\rho_{23}+\zeta_3\sqrt{1-\rho_{13}^2-\rho_{23}^2}\right).
\end{align}
Then, we write the evolution under the risk-neutral measure $\Q$ as:
\begin{align}
dX^{(1)}_t &= \kappa\left(X^{(2)}_t+X^{(3)}_t-X^{(1)}_t-\lambda_1/\kappa\right)dt\notag\\
&+\sigma_1 dZ^{\Q,1}_t,
\end{align}
\begin{align}
%dX^{(1)}_t &= \kappa\left(X^{(2)}_t+X^{(3)}_t-X^{(1)}_t-\lambda_1/\kappa\right)dt\notag\\
%&+\sigma_1 dZ^{\Q,1}_t,\\
dX^{(2)}_t &= \frac{1}{\epsilon}\left(\alpha_2-X^{(2)}_t-\epsilon\lambda_2\right)dt\notag\\
&+\frac{1}{\sqrt \epsilon}\sigma_2 \left(\rho_{12} dZ^{\Q,1}_t+\sqrt{1-\rho_{12}^2} dZ^{\Q,2}_t\right),\\
dX^{(3)}_t &= \delta\left(\alpha_3-X^{(3)}_t-\lambda_3/\delta\right)dt+\sqrt{\delta}\sigma_3\bigg(\rho_{13} dZ^{\Q,1}_t\notag\\
&+\rho_{23}dZ^{\Q,2}_t+\sqrt{1-\rho_{13}^2-\rho_{23}^2}dZ^{\Q,3}_t\bigg).
\end{align}

For convenience, we define
\begin{align}
\bm{X}_t &= (X^{(1)}_t,X^{(2)}_t,X^{(3)}_t)',\\
\bm{Z}_t^\P &= (Z^{(\P,1)}_t,Z^{(\P,2)}_t,Z^{(\P,3)}_t)',\\
\bm{Z}_t^\Q &= (Z^{(\Q,1)}_t,Z^{(\Q,2)}_t,Z^{(\Q,3)}_t)',\\
\bm{\mu} &= (0,\alpha_2/\epsilon,\delta\alpha_3)',\\
\bm{\lambda} &= (\lambda_1,\lambda_2,\lambda_3)',\\
\bm{K} &= \bcm \kappa & -\kappa &-\kappa\\0 &1/\epsilon &0\\0&0&\delta
\ecm,\\
\bm{\Sigma}&=\bcm \sigma_1&0&0\\0&\sigma_2/\sqrt{\epsilon}&0\\0&0&\sqrt{\delta}\sigma_3\ecm, 
\end{align}
and
\begin{equation}
\bm{C} = \bcm 1 &0&0\\\rho_{12}&\sqrt{1-\rho_{12}^2}&0\\\rho_{13}&\rho_{23}&\sqrt{1-\rho_{13}^2-\rho_{23}^2}\ecm.
\end{equation}
Then, the evolution for $\bm{X}_t$ under   measures $\mathbb{P}$ and $\mathbb{Q}$ can be written concisely   as
\begin{equation}
d\bm{X}_t = (\bm{\mu}- \bm{K}\bm{X}_t) dt + \bm{\Sigma C}d\bm{Z}_t^{\P},\label{N_Factor_SDE_P}
\end{equation}
and
\begin{equation}
d\bm{X}_t=\bigg(\bm{\mu}-\bm{\lambda}-\bm{K} \bm{X}_t\bigg) dt + \bm{\Sigma C}d\bm{Z}_t^{\Q}.\label{N_Factor_SDE_Q}
\end{equation}

\begin{remark}
If the stochastic mean of log price $X^{(1)}$ is only modeled by $X^{(2)}$ or $X^{(3)}$, instead of their sum, it will reduce to the CTOU model, which is used  in \cite{MenciaSentana} for pricing VIX
futures. Under this model,    the futures portfolio optimization problem has been studied in \cite{firstpaper}. 
\end{remark}

\section{Futures Pricing and Futures Trading}\label{sect-formula-multiscale}

\subsection{Futures Pricing}
Let us consider three futures contracts $F^{(1)}$, $F^{(2)}$ and $F^{(3)}$, written on the same underlying asset $S$ with three arbitrarily chosen maturities $T_1$, $T_2$ and $T_3$ respectively. Recall that the asset price is given by \[S_t = \exp(X_t^{(1)}), \quad t\ge 0.\] Then, the  futures price at time $t\in[0,T]$ is given by
\begin{equation}
F^{(k)}(t,\bm{x}):= \E^\Q\big[\exp(X^{(1)}_{T_k})\,|\,\bm{X}_t=\bm{x}\,\big],\label{Futures_Definition}
\end{equation}
for $k=1,2,3$.
Define the linear differential  operator
\begin{equation}
\mathcal{L}^\Q\cdot = \bigg(\bm{\mu}-\bm{\lambda}-\bm{K} \bm{x}\bigg)'\nabla_{\bm{x}}\cdot+\frac{1}{2}\Tr \bigg(\bm{\Sigma\Omega  \Sigma} \nabla_{\bm{xx}}\cdot \bigg), \label{L_Q}
\end{equation}
where $\nabla_{\bm{x}}\cdot = ( \partial_{x_1}\cdot,\cdots ,
\partial_{x_N}\cdot)'$ is the nabla operator and Hessian operator $\nabla_{\bm{xx}}\cdot$ satisfies
\begin{equation}
\nabla_{\bm{xx}} \cdot=
\begin{bmatrix}
    \partial_{x_1}^2\cdot &  \partial_{x_1x_2}\cdot  & \dots  & \partial_{x_1x_N}\cdot  \\
    \partial_{x_1x_2}\cdot  & \partial_{x_2}^2\cdot   & \dots  & \partial_{x_2x_N}\cdot  \\
    \vdots & \vdots  & \ddots & \vdots \\
    \partial_{x_1x_N}\cdot  & \partial_{x_2x_N}\cdot   & \dots  & \partial_{x_N}^2\cdot 
\end{bmatrix}.
\end{equation}
Then, for $k=1,2,3$, the   futures price function $F^{(k)}(t,\bm{x})$ solves the following PDE
\begin{equation}
(\partial_t  + \mathcal{L}^\Q) F^{(k)}(t,\bm{x}) = 0,\label{Futures_PDE}
\end{equation}
for $(t,\bm{x})\in[0,T)\times \R^N$, with the terminal condition $F^{(k)}(T,\bm{x})=\exp(\bm{e_1}'\bm{x})$ for $\bm{x}\in \R^N$, where $\bm{e_1}=(1, 0 ,  0)'$. 

\begin{proposition}\label{prop-futures-price}
The futures price $F^{(k)}(t,\bm{x})$ is given by
\begin{equation}\label{Futures_Price_MultiScale}
F^{(k)}(t,\bm{x}) = \exp{\bigg(\bm{a}^{(k)}(t)'\bm{x}+\beta^{(k)}(t)\bigg)},
\end{equation}
where $\bm{a}^{(k)}(t)$ and $\beta^{(k)}(t)$ satisfy
\begin{align}\label{a(t)}
\bm{a}^{(k)}(t)&=\exp{\bigg(-(T-t)\bm{K}'\bigg)}\bm{e_1}\\
\beta^{(k)}(t) &= \int_t^T (\bm{\mu}-\bm{\lambda})'\bm{a}^{(k)}(s)\notag\\
&+\frac{1}{2}\Tr\bigg(\bm{\Sigma\Omega\Sigma}\bm{a}^{(k)}(s)\bm{a}^{(k)}(s)'\bigg)ds.\label{beta(t)}
\end{align}
\end{proposition}

\begin{proof}
We substitute the ansatz solution \eqref{Futures_Price_MultiScale} into PDE \eqref{Futures_PDE}. The $t$-derivative is given by
\begin{align}\label{F_t}
\partial_t F(t,\bm{x}) &= \left( \left(\frac{d\bm{a}^{(k)}(t)}{dt}\right)'\bm{x}+\frac{d\beta^{(k)}(t)}{dt}\right)\notag\\
&\times\exp{\bigg(\bm{a}^{(k)}(t)'\bm{x}+\beta^{(k)}(t)\bigg)}.
\end{align} Then, the first and second derivatives satisfy
\begin{align}\label{F_x}
\nabla_{\bm{x}} F(t,\bm{x}) &= \bm{a}^{(k)}(t)F(t,\bm{x}),\\
\label{F_xx}
\nabla_{\bm{xx}} F(t,\bm{x}) &= \bm{a}^{(k)}(t)\bm{a}^{(k)}(t)'F(t,\bm{x}).
\end{align}
By substituting \eqref{F_t}, \eqref{F_x}, and \eqref{F_xx} into PDE (\ref{Futures_PDE}), we obtain
\begin{equation}
\frac{d\bm{a}^{(k)}(t)}{dt} -\bm{K}'\bm{a}^{(k)}(t) = 0,\label{a_ODE}
\end{equation}
and
\begin{align}
\frac{d\beta^{(k)}(t)}{dt}+(\bm{\mu}-\bm{\lambda})'\bm{a}^{(k)}(t)\notag\\
+\frac{1}{2}\Tr\bigg(\bm{\Sigma\Omega\Sigma}\bm{a}^{(k)}(t)\bm{a}^{(k)}(t)'\bigg)=0.\label{beta_ODE}
\end{align}
The terminal conditions of $\bm{a}^{(k)}(t)$ and $\beta^{(k)}(t)$ are given by 
\begin{equation}
\bm{a}^{(k)}(T)=\bm{e_1},\quad \beta^{(k)}(T) = 0.\label{a_beta_terminal_conditions}
\end{equation}

By direct substitution, the solutions to   ODEs (\ref{a_ODE}) and (\ref{beta_ODE})  are given by \eqref{a(t)} and \eqref{beta(t)}. 
\end{proof}

\subsection{Dynamic Futures Portfolio}
Now consider a collection of $M$ contracts of different maturities available to trade, where $M=1,$ 2, 3. We note that there are only three sources of randomness, so trading three contracts is sufficient. Any additional contract would be redundant in this model. By Proposition \ref{prop-futures-price}, we have
\begin{align}
\frac{dF_t^{(k)}}{F_t^{(k)}}&=\bm{a}^{(k)}(t)'\bm{\lambda}dt+\bm{a}^{(k)}(t)'\bm{\Sigma C} d\bm{Z}_t^{\P}\\
&\equiv \mu^{(k)}_F(t) dt + \bm{\sigma}^{(k)}_F(t)' d\bm{Z}_t^{\P},\label{mu_F^k}
\end{align}
where we have defined
\begin{equation}
\mu^{(k)}_F(t) \equiv \bm{a}^{(k)}(t)'\bm{\lambda},\quad\bm{\sigma}^{(k)}_F(t) \equiv \bm{C}'\bm{\Sigma}'\bm{a}^{(k)}(t).
\label{mu_F_Sigma_F_Definition_m}
\end{equation}
Define 
\begin{equation}
d\bm{F}_t = \left( \frac{dF^{(1)}_t}{F^{(1)}_t}, \cdots, \frac{dF^{(M)}_t}{F^{(M)}_t} \right)'.
\end{equation} Then, in matrix form, the system of futures dynamics is given by the set of SDE:
\begin{equation}
  d\bm{F}_t = \bm{\mu_F}(t) dt+\bm{\Sigma_F}(t)d\bm{Z}_t^\P ,
\end{equation}
where
\begin{align}\label{mu_F_Sigma_F_matrix_Definition_m}
\bm{\mu_F}(t) &= \left( \mu_F^{(1)}(t), \cdots , \mu_F^{(M)}(t) \right) '  \\
 \bm{\Sigma_F}(t) &= \left(\bm{\sigma}^{(1)}_F(t),\cdots, \bm{\sigma}^{(M)}_F(t)  \right)'.
\end{align}

Here, we assume there be no redundant futures contract, which means any futures contract could not be replicated by other $M-1$ futures contracts, indicating that $\rank\left(\bm{\Sigma_F}\right)=M$. 
%Then, we define $\bm{A}(t)=(\bm{a}^{(1)}(t),\ldots,\bm{a}^{(M)}(t))$ as an $N\times M$ matrix and it leads to $\bm{\mu_F}(t)=\bm{A}'(t)\bm{\lambda}$ and 

%If $\rank(\bm{\Sigma_F})=l<M$, then we only need to consider dynamically trading $l$ futures contracts instead of $m$ futures contracts. It is because these two problems are equivalent if $\rank(\bm{\Sigma_F})=l<M$. 

Next, we consider the trading problem for the investor. Let strategy $\bm{\pi}_t=\left(\pi^{(1)}_t,\cdots,\pi^{(M)}_t\right)'$, 
where the element $\pi^{(k)}_t$ 
denotes the amount of money invested in $k$-th futures contract.
In addition, we assume the interest rate be zero for simplicity. Then, for any admissible strategy $\pi$, the wealth process is 
\begin{align}
dW_t^\pi &= \sum_{k=1}^M \pi^{(k)}_t \frac{dF^{(k)}_t}{F^{(k)}_t}\notag\\
&= \bm{\pi}_t' \bm{\mu_F}(t) dt+\bm{\pi}_t'\bm{\Sigma_F}(t) d\bm{Z}_t^\P. \label{Wealth_SDE_m}
\end{align}
We note that the wealth process is only determined by the strategy $\bm{\pi}_t$ and it is not affected by factors variable $\bm{X}$ and futures prices $\bm{F}$.

The investor's risk preference is described by the exponential utility: 
\begin{equation}
 U(w)= -\exp(-\gamma w),\label{exponential_utility}
 \end{equation} 
where $\gamma>0$ denotes the coefficient of risk aversion. A  strategy $\bm{\pi}$ is said to be admissible if $\bm{\pi}$ is real-valued 
  progressively measurable and satisfies the Novikov condition  \cite{Novikov1972}:
\begin{equation}\label{integrability_condition}
\E^{\P}\bigg[\exp\left(\int_t^{\tilde T}\frac{\gamma^2}{2} \bm{\pi}_s'\bm{\Sigma_F}(s) \bm{\Sigma_F}'(s)\bm{\pi}_s ds \right)\bigg]< \infty.
\end{equation}

The investor fixes a finite optimization horizon $0< \tilde T \le T_1$, which means that $\tilde T$ has to be less than or equal to the maturity of the earliest expiring contract,
and seeks an admissible strategy $\bm{\pi}$ that maximizes the expected utility of wealth at $\tilde T$:
\begin{equation}
u(t,w)=\sup_{\bm{\pi}\in \mathcal{A}_t} \E[U(W_{\tilde T}^{\bm{\pi}})|W_t=w],
\label{Value_Function_m}
\end{equation}
where $\mathcal{A}_t$ denotes the set of admissible controls at the initial time $t$. Since the wealth SDE (\ref{Wealth_SDE_m}) does not depend on the factors variable $\bm{X}$ and futures prices $\bm{F}$, the value function does not depend on them either. 

To facilitate presentation, we define
\begin{equation}
\mathcal{L}^{\bm{\pi}}\cdot = \bm{\pi}_t'\bm{\mu_F}(t) \partial_w\cdot+\frac{1}{2} \bm{\pi}_t'\bm{\Sigma_F}(t) \bm{\Sigma_F}'(t)\bm{\pi}_t\partial_{ww}\cdot.
\end{equation}
Then, following the standard verification approach to dynamic programming \cite{FlemingSoner93}, the candidate value function   $u(t,w)$ and optimal trading strategy $\pi^*$ is found from the Hamilton-Jacobi-Bellman (HJB) equation
\begin{equation}
\partial_t u + \sup_{\bm{\pi}} 
\mathcal{L}^{\bm{\pi}} u 
= 0, \label{Value_Function_HJB_m}
\end{equation}
for $(t,w)\in[0,\tilde T)\times \mathbb{R}$, along with the terminal condition $u(T,w)=-e^{-\gamma w}$, for $w\in  \mathbb{R}.$  

\begin{theorem}\label{thm2} The unique solution to the HJB equation \eqref{Value_Function_HJB_m} is given by
\begin{equation}
u(t,w)=-\exp\left(
-\gamma w -\frac{1}{2}\int_t^{\tilde T}\Lambda^2(s)ds\right),
\label{u_Candidate_Solution_m}
\end{equation}
where   
\begin{equation}\label{Lambda}
 \Lambda^2(t)=\bm{\mu_F}(t)' \left(\bm{\Sigma_F}(t)\bm{\Sigma_F}(t)'\right)^{-1} \bm{\mu_F}(t)\,.
 \end{equation} 
 The optimal futures trading strategy is explicitly given by
\begin{equation}
\bm{\pi}^*(t)=\frac{1}{\gamma}
\left( \bm{\Sigma_F}(t)\bm{\Sigma_F}'(t) \right)^{-1}\bm{\mu_F}(t)
.\label{Optimal_Strategy_m}
\end{equation}
 
\end{theorem}

\begin{proof}
We will first use the ansatz
\begin{equation}
u(t,w)= -e^{-\gamma w} h(t).
\end{equation}
Then, using the relations
\begin{equation}
\partial_t u = -e^{-\gamma w} \partial_t h(t),\quad \partial_w u=\gamma e^{-\gamma w} h(t), 
\end{equation}
and
\begin{equation}
\partial_{ww} u=-\gamma^2 e^{-\gamma w} h(t),
\end{equation}
the PDE \eqref{Value_Function_HJB_m} becomes
\begin{align}
-\frac{d}{dt} h(t) + &\sup_{\bm{\pi}_t} \bigg[\gamma
\bm{\pi'}_t\bm{\mu_F} (t) h
\notag\\
&-\frac{1}{2} \gamma^2 \bm{\pi}_t'\bm{\Sigma_F}(t)\bm{\Sigma_F}'(t)\bm{\pi}_t h \bigg]
= 0, \label{h_HJB_m}
\end{align}
with terminal condition $h(\tilde T)=1$.
From the first-order condition, which is obtained from differentiating the terms inside the supremum with respect to $\bm{\pi}_t$ and setting the equation to zero, we have
\begin{equation}
\gamma \bm{\mu_F}(t) 
-\gamma^2 \bm{\Sigma_F}(t)\bm{\Sigma_F}'(t) \bm{\pi}_t=0.
\end{equation}
Recall that $\rank(\bm{\Sigma_F}(t))=M$. Then, $\bm{\Sigma_F}(t)\bm{\Sigma_F}'(t)$ is an $M\times M$ invertible matrix. Accordingly, we have the optimal strategy \eqref{Optimal_Strategy_m}.
Given the fact that $A'A$ is the semi-positive definite matrix for any matrix $A$, the time-dependent component $\Lambda^2(t)=\bm{\mu_F}(t)' (\bm{\Sigma_F}(t)\bm{\Sigma_F}(t)')^{-1} \bm{\mu_F}(t)$ is always non-negative.

 Substituting $\bm{\pi}^*$ back, the equation \eqref{h_HJB_m} becomes
\begin{align}
-\frac{d}{dt} h(t) &+ 
\frac{1}{2}\Lambda^2(t) h(t)=0.
\label{h_HJB_Reduced_m}
\end{align}
Accordingly, we have
\begin{equation}
h(t)=\exp{\bigg(-\frac{1}{2}\int_t^{\tilde T}\Lambda^2(s)ds\bigg)}.
\end{equation}
\end{proof}

\begin{example}
  If there is only one futures contract $F^{(1)}$ available in the market. Then by \eqref{mu_F_Sigma_F_matrix_Definition_m}, we have
\begin{equation}\label{3factor-single-futures-mu}
 \bm{\mu_F} =  \mu_F^{(1)},\quad \bm{\Sigma_F} =  \bm{\sigma}_F^{(1)}(t)'.
\end{equation} 
Then, the optimal strategy \eqref{Optimal_Strategy_m} becomes
\begin{align}\label{Optimal_Strategy_MultiScale_Port1}
\pi^{(1)*}(t) &= \frac{1}{\gamma}\frac{\bm{\mu_F}(t)}{\bm{\Sigma_F}(t)\bm{\Sigma_F}'(t)} =  \frac{1}{\gamma}\frac{\mu_F^{(1)}}{\bm{\sigma}_F^{(1)}(t)'\bm{\sigma}_F^{(1)}(t)},
\end{align}
where $\mu_F^{(1)}$ and $\bm{\sigma}_F^{(1)}(t)$ are shown as \eqref{mu_F^k}.
\end{example}

In order to quantify  the value of trading futures to the investor, we define the investor's certainty equivalent  associated with the utility maximization problem. The certainty equivalent is the guaranteed cash amount   that would yield the same  utility as that from dynamically trading futures according to \eqref{Value_Function_m}. This amounts to applying the inverse of the utility function to the value function in \eqref{u_Candidate_Solution_m}. Precisely, we define
\begin{align}
\mathcal{C}(t,w)&:=U^{-1}(u(t,w))\\
&=w+\frac{1}{2\gamma}\int_t^{\tilde T}\Lambda^2(s)ds.\label{nfactor:certequiv}
\end{align}

Therefore, the certainty equivalent is the sum of the investor's wealth $w$ and a non-negative time-dependent component  $\frac{1}{2\gamma}\int_t^{\tilde T}\Lambda^2(s)ds$.  The certainty equivalent is also inversely proportional to the risk aversion parameter $\gamma$, which means  that a more risk averse investor has a lower certainty equivalent, valuing the futures trading opportunity less. From \eqref{a(t)}, \eqref{mu_F_Sigma_F_Definition_m}, \eqref{mu_F_Sigma_F_matrix_Definition_m} and \eqref{Lambda}, we see that the certainty equivalent   depends on the constant matrix $\bm{K}$, volatility matrix $\bm{\Sigma}$, correlation matrix $\bm{C}$ and market prices of risk $\bm{\lambda}$. Nevertheless, the certainty equivalent does not depend on the current values of factors $\bm{X}_t$.

\begin{table}
\vspace*{.2cm}
\centering
\begin{scriptsize}
\begin{tabular}{c|c|c|c|c}
\hline
$X^{(1)}_0$ & $X^{(2)}_0$& $X^{(3)}_0$& $\alpha_2$& $\alpha_3$ \\
\hline
1& 0.5 &0.5 &0.5 & 0.5\\
\hline
\hline
 $\epsilon$& $\delta$ & $\sigma_1$& $\sigma_2$&$\sigma_3$ \\
\hline
0.05& 0.01 & 0.8  & 0.02 & 0.3\\
\hline
\end{tabular}

\vspace*{.1 cm}

\begin{tabular}{c|c|c|c|c|c}

\hline

 $\rho_{12}$&$\rho_{13}$ & $\rho_{23}$ & $\lambda_1$ & $\lambda_2$ & $\lambda_3$\\
 
\hline
 0& 0& 0 & 0.02 & 0.02 & 0.02 \\
 \hline
 \hline
 $T_1$&$T_2$&$T_3$ &$\tilde{T}$ & $\gamma$& $\kappa$\\
\hline
$1/12$&$2/12$&$3/12$& $1/12$ & $1$&5\\
\hline
\end{tabular}
\end{scriptsize}

\caption{Parameters for the MCTOU model.}
\label{Parameters_Table}
\end{table}

\begin{figure}
    \centering
%    {\includegraphics[width=0.45\textwidth]{Factors}}
{\includegraphics[trim=0.4cm   1.5cm  1.5cm  2.5cm,clip, width=0.50\textwidth]{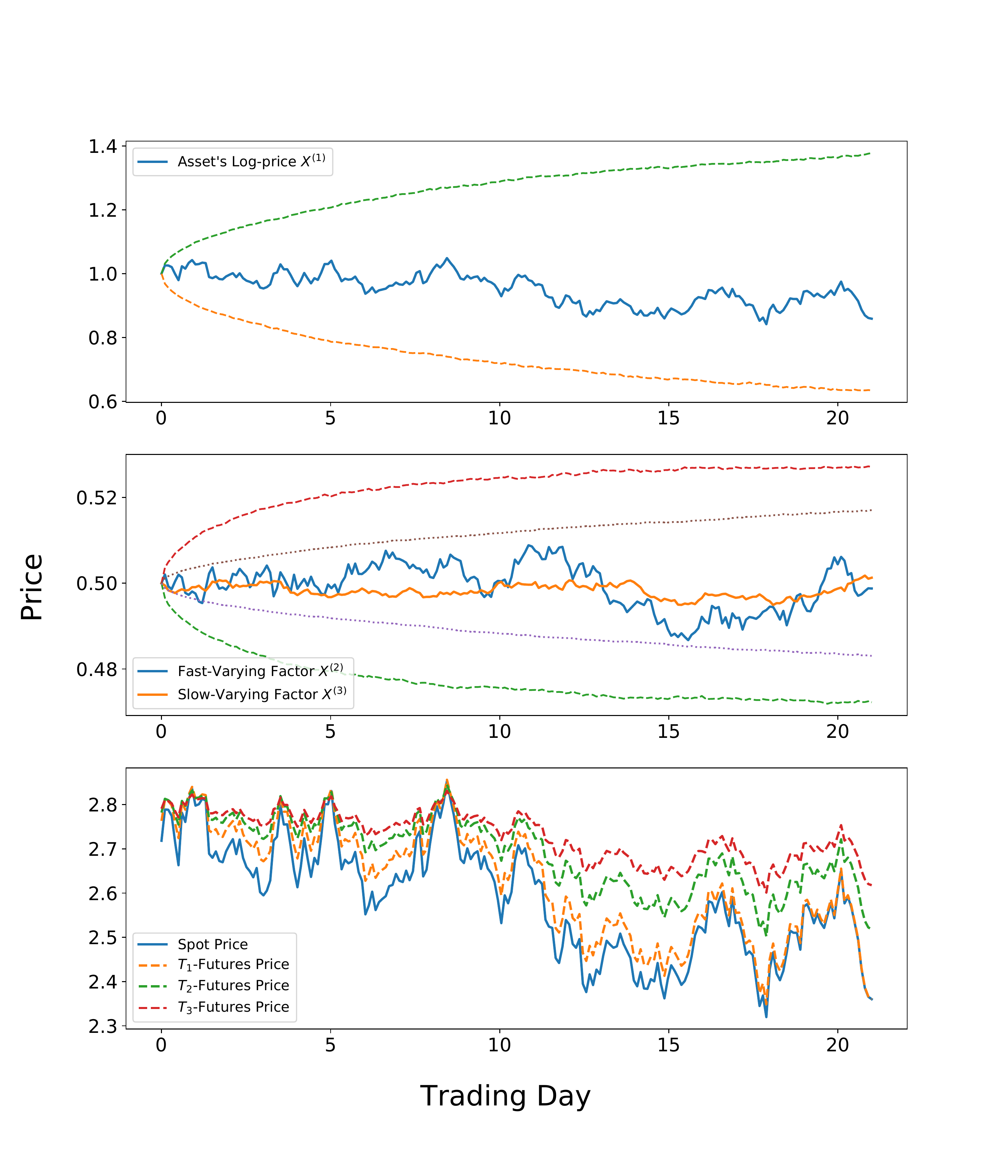}}
\caption{
Top: simulation paths for asset's log-price $X^{(1)}$. Dashed curves represent $95\%$ confidence interval. Middle: simulation path for fast varying factor $X^{(2)}$ and slow-varying factor $X^{(3)}$. Dashed and dotted curves represent $95\%$ confidence interval for $X^{(2)}$ and $X^{(3)}$, respectively. Bottom: Sample price paths for the underlying asset and associated futures.}
\label{Prices}
\end{figure}

\section{Numerical Illustration}\label{sect-numeric}
In this section, we simulate the MCTOU process and illustrate the outputs from our trading model.   With the closed-form expressions obtained in the Section \ref{sect-formula-multiscale},  we now generate the futures prices, optimal strategies and wealth processes numerically, using the parameters in Table  \ref{Parameters_Table}. Primarily, we let $\epsilon$ and $\delta$ be small parameters and we consider trading three futures with maturities $T_1=1/12$ year, $T_2=2/12$ year and $T_3=3/12$ year.   Then, our trading horizon will be $\tilde{T}=1/12$ year, no greater than the futures maturities. We assume 252 trading days in a year and 21 trading days in a month (or 1/12 year). In our figures, we show the corresponding trading days on the $x$ axis.

\begin{figure}
    \centering
{\includegraphics[trim=0.4cm   1.8cm  1.5cm  2.2cm,clip, width=0.48\textwidth]{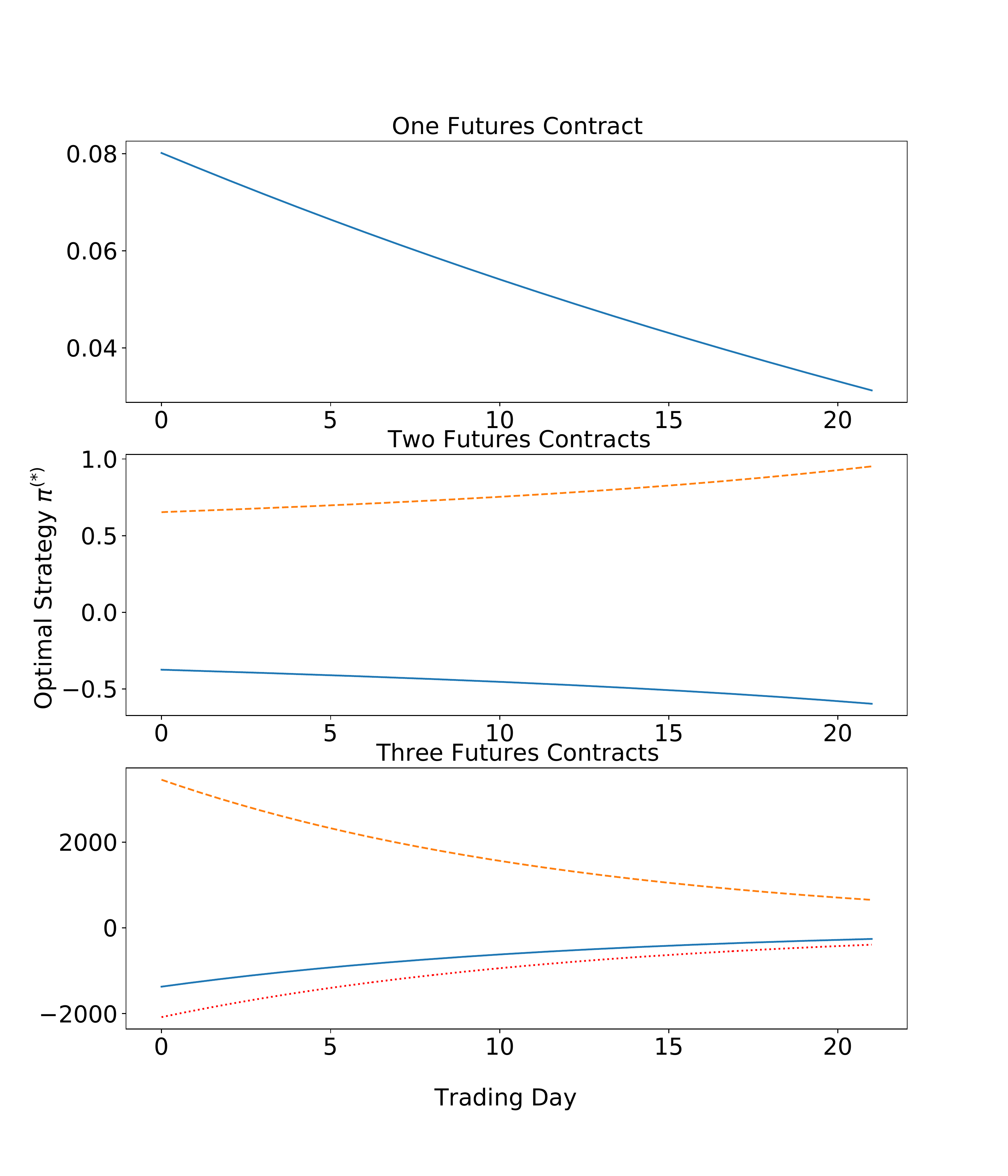}
}
\caption{
Optimal strategies $\pi^{(*)}$ for different futures combinations. Solid, dashed and dotted lines represent the optimal position (in \$) on $T_1$-futures, $T_2$-futures and $T_3$-futures, respectively.}
\label{Optimal_Investment_MultiScale}
\end{figure}

\begin{table*}
\vspace*{.2cm}
\centering\begin{scriptsize}
  \begin{tabular}{cc|ccccccc}
     \hline
     \multicolumn{2}{c|}{\multirow{2}{*}{Parameters}} &\multicolumn{7}{c}{Futures Combinations (Maturity)}\\
     \cline{3-9}
      & & $T_1$ & $T_2$ & $T_3$ & $\{T_1, T_2\}$ & $\{T_1, T_3\}$ & $\{T_2, T_3\}$ & $\{T_1, T_2,T_3\}$ \\
      \hline
\hline
\multirow{3}{*}{$\rho_{12}= 0 $} 
& $\rho_{13}= -0.5 $ &
0.563 & 1.58 & 3.25 & 5.36 & 4.65 & 4.41 & 419 \\
& $\rho_{13}= 0 $ &
0.502 & 1.09 & 1.74 & 3.09 & 2.62 & 2.41 & 417 \\
& $\rho_{13}= 0.5 $ &
0.456 & 0.837 & 1.19 & 2.88 & 2.35 & 2.01 & 417 \\

\hline
\multirow{3}{*}{$\rho_{12}= 0.5 $} 
& $\rho_{13}= -0.5 $ &
0.561 & 1.56 & 3.23  & 5.34  & 4.64  & 4.40  & 543 \\
& $\rho_{13}= 0 $ &
0.500 & 1.08  & 1.73  & 3.08  & 2.62  & 2.40  & 542 \\
& $\rho_{13}= 0.5 $ &
0.454 & 0.833 & 1.18  & 2.87  & 2.34  & 2.01  & 541 \\

\hline
\multirow{3}{*}{$\rho_{12}= -0.5 $} 
& $\rho_{13}= -0.5 $ &
0.565 & 1.59  & 3.27  & 5.39  & 4.66  & 4.42  & 571 \\
& $\rho_{13}= 0 $ &
0.504 & 1.10  & 1.75  & 3.11  & 2.63  & 2.41  & 569 \\
& $\rho_{13}= 0.5 $ &
0.457 & 0.842 & 1.20  & 2.90  & 2.36  & 2.02  & 569 \\

     \hline
   \end{tabular}\end{scriptsize}
     \caption{Certainty equivalents ($\times 10^{-4}$) for all possible futures combinations under different correlations.\label{CE_Table}}
\end{table*}

In Figure \ref{Prices}, we plot the simulation paths and $95\%$ confidence intervals for three factors in the top figure and middle figure. As shown in the middle panel, the  $95\%$ confidence interval of the slow-varying factor $X^{(3)}$ is much narrower than the one for fast-varying factor $X^{(2)}$. At the bottom, we plot the spot price and futures prices. The three paths for the futures prices are highly correlated and $T_1$-futures price is equal to the asset's spot price at its maturity date $T_1$, which is the 21st trading day.

In Figure \ref{Optimal_Investment_MultiScale}, we plot the optimal strategies as functions of time for different portfolios and different correlation parameters. In each sub-figure, from top to bottom, we show the optimal strategies for one-contract portfolio, two-contract portfolio and three-contract portfolio respectively. The optimal investments on $T_1$-futures, dashed lines represent the optimal investment on $T_1$-futures, $T_2$-futures, and $T_3$-futures are represented by solid, dashed, and dotted lines respectively. 
The optimal cash amount invested are deterministic functions for time, but the optimal units of futures held do vary continuously with the prevailing futures price. 

Moreover, the investor takes large long/short positions in three-contract portfolio since all sources of risk can be hedged. We provide sample path for wealth process for three-contract portfolios in the Figure \ref{Wealth}.

In  Figure \ref{nplotmn}, we  see that  the certainty equivalent increases as a function of trading horizon $\tilde T$, which means that the more time the investor has, the more valuable is the trading opportunity. As the trading horizon reduces to zero, the certainty equivalent converges to the initial wealth $w$, which is set to be 0 in this example, as expected from \eqref{nfactor:certequiv}. Also, with a lower risk aversion parameter $\gamma$, the investor has a higher certainty equivalent for any given trading horizon.

 Table \ref{CE_Table} shows the certainty equivalents for all possible futures combinations under various correlation configurations. The certainty equivalent is much higher when   more contracts are traded. In addition, if there is only one futures contract to trade, the certainty equivalent is increasing with respect to its maturity, see first three columns. The certainly equivalents tend to be higher when $\rho_{12}$ and $\rho_{13}$ are negative. 

\begin{figure}[h]
    \centering
{\includegraphics[trim=0.3cm   0cm  1.5cm  1.2cm,clip,  width=0.45\textwidth]{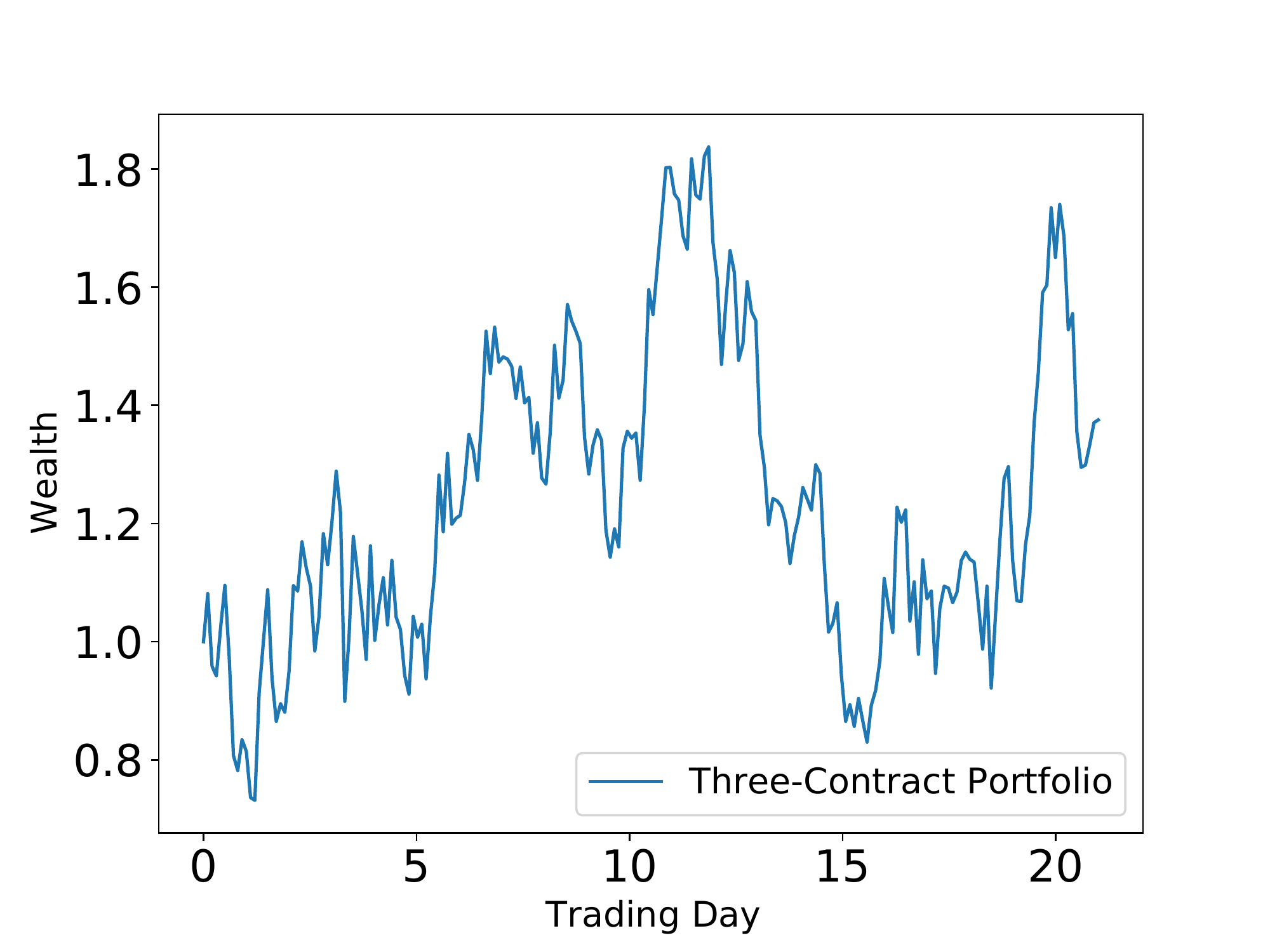}}
\caption{
Sample path for wealth process for the three-futures portfolio. }
\label{Wealth}
\end{figure}

\begin{figure}[h]
    \centering
{\includegraphics[trim=1cm   0.7cm  1.8cm  2cm,clip, width=0.45\textwidth]{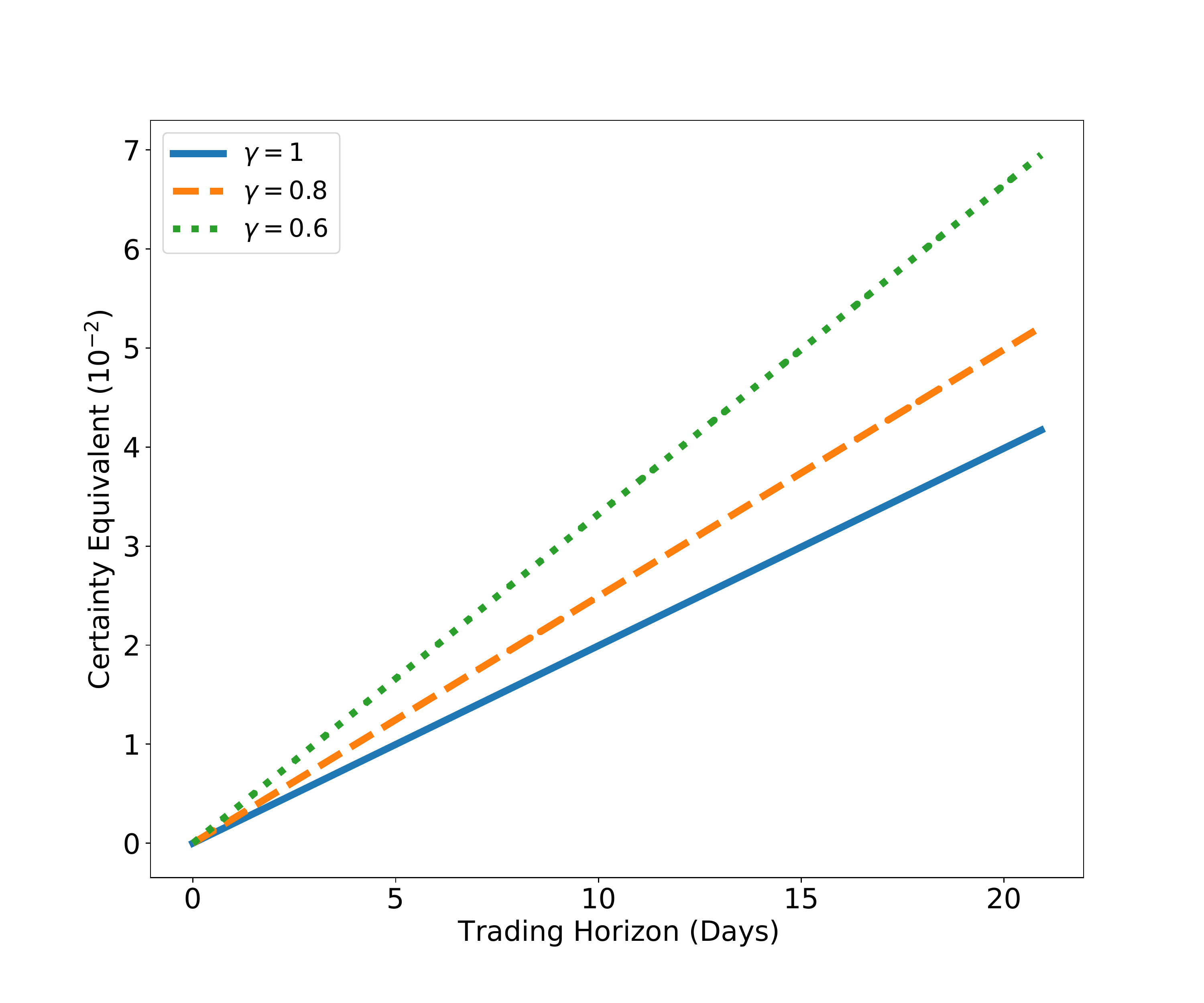}}
\caption{
Certainty equivalents for the three-futures portfolio as  the  trading horizon $\tilde T$  and   risk aversion parameter $\gamma$ vary.}
\label{nplotmn}
\end{figure}

\section{Conclusion}\label{sect-conclusion}
We have studied the optimal trading of   futures under a multiscale multifactor model. Closed-form expressions for the optimal controls and   value function are derived through the analysis of the associated HJB equation. Using these, we have illustrated the path behaviors of the  futures prices and optimal positions. We also quantify the values of the trading different combinations of futures under different model parameters.

\addtolength{\textheight}{-12cm}   % This command serves to balance the column lengths

%%%%%%%%%%%%%%%%%%%%%%%%%%%%%%%%%%%%%%%%%%%%%%%%%%%%%%%%%%%%%%%%%%%%%%%%%%%%%%%%

\bibliographystyle{IEEEtran}

\bibliography{biblio}

\begin{thebibliography}{10}
\providecommand{\url}[1]{#1}
\csname url@rmstyle\endcsname
\providecommand{\newblock}{\relax}
\providecommand{\bibinfo}[2]{#2}
\providecommand\BIBentrySTDinterwordspacing{\spaceskip=0pt\relax}
\providecommand\BIBentryALTinterwordstretchfactor{4}
\providecommand\BIBentryALTinterwordspacing{\spaceskip=\fontdimen2\font plus
\BIBentryALTinterwordstretchfactor\fontdimen3\font minus
  \fontdimen4\font\relax}
\providecommand\BIBforeignlanguage[2]{{%
\expandafter\ifx\csname l@#1\endcsname\relax
\typeout{** WARNING: IEEEtran.bst: No hyphenation pattern has been}%
\typeout{** loaded for the language `#1'. Using the pattern for}%
\typeout{** the default language instead.}%
\else
\language=\csname l@#1\endcsname
\fi
#2}}

\bibitem{Fouque2000}
J.-P. Fouque, G.~Papanicolaou, and R.~Sircar, \emph{Derivatives in Financial
  Markets with Stochastic Volatility}.\hskip 1em plus 0.5em minus 0.4em\relax
  Cambridge University Press, 2000.

\bibitem{nfactor}
G.~Cortazar and L.~Naranjo, ``An {N}-factor {G}aussian model of oil futures
  prices,'' \emph{Journal of Futures Markets}, vol.~26, no.~3, pp. 243--268,
  2006.

\bibitem{CORTAZAR2017182}
G.~Cortazar, M.~Lopez, and L.~Naranjo, ``A multifactor stochastic volatility
  model of commodity prices,'' \emph{Energy Economics}, vol.~67, pp. 182--201,
  2017.

\bibitem{Merton}
R.~Merton, ``Optimum consumption and portfolio rules in a continuous time
  model,'' \emph{Journal of Economic Theory}, vol.~3, no.~4, pp. 373--413,
  1971.

\bibitem{firstpaper}
T.~Leung and R.~Yan, ``Optimal dynamic pairs trading of futures under a
  two-factor mean-reverting model,'' \emph{International Journal of Financial
  Engineering}, vol.~5, no.~3, p. 1850027, 2018.

\bibitem{secondpaper}
------, ``{A stochastic control approach to managed futures portfolios},''
  \emph{International Journal of Financial Engineering}, vol.~6, no.~1, p.
  1950005, 2019.

\bibitem{Leung2019}
T.~Leung and Y.~Zhou, ``Dynamic optimal futures portfolio in a regime-switching
  market framework,'' \emph{Internation Journal of Financial Engineering},
  vol.~6, no.~4, p. 1950034, 2019.

\bibitem{BahmanLeung}
B.~Angoshtari and T.~Leung, ``Optimal dynamic basis trading,'' \emph{Annals of
  Finance}, vol.~15, no.~3, pp. 307--335, 2019.

\bibitem{Angoshtari2019b}
------, ``Optimal trading of a basket of futures contracts,'' \emph{Annals of
  Finance}, 2020, published online.

\bibitem{LeungLiLiZheng2015}
T.~Leung, J.~Li, X.~Li, and Z.~Wang, ``Speculative futures trading under mean
  reversion,'' \emph{Asia-Pacific Financial Markets}, vol.~23, no.~4, pp.
  281--304, 2016.

\bibitem{meanreversionbook2016}
T.~Leung and X.~Li, \emph{Optimal Mean Reversion Trading: Mathematical Analysis
  and Practical Applications}.\hskip 1em plus 0.5em minus 0.4em\relax World
  Scientific, Singapore, 2016.

\bibitem{LeungWard}
T.~Leung and B.~Ward, ``The golden target: analyzing the tracking performance
  of leveraged gold {ETF}s,'' \emph{Studies in Economics and Finance}, vol.~32,
  no.~3, pp. 278--297, 2015.

\bibitem{MenciaSentana}
J.~Mencia and E.~Sentana, ``Valuation of {VIX} derivatives,'' \emph{Journal of
  Financial Economics}, vol. 108, pp. 367--391, 2013.

\bibitem{Novikov1972}
A.~A. Novikov, ``On an identity for stochastic integrals,'' \emph{Theory of
  Probability \& Its Applications}, vol.~17, no.~4, 1972.

\bibitem{FlemingSoner93}
W.~H. Fleming and H.~M. Soner, \emph{Controlled Markov Processes and Viscosity
  Solutions}.\hskip 1em plus 0.5em minus 0.4em\relax Springer-Verlag, 1993.

\end{thebibliography}

\end{document}